\documentclass[sigconf,nonacm=true,screen=true]{acmart}

 \usepackage{balance}
\def\BibTeX{{\rm B\kern-.05em{\sc i\kern-.025em b}\kern-.08emT\kern-.1667em\lower.7ex\hbox{E}\kern-.125emX}}
    
\copyrightyear{2020} 
\acmYear{2020} 
\setcopyright{acmlicensed}\acmConference[SPAA '20]{Proceedings of the 32nd ACM Symposium on Parallelism in Algorithms and Architectures}{July 15--17, 2020}{Virtual Event, USA}
\acmBooktitle{Proceedings of the 32nd ACM Symposium on Parallelism in Algorithms and Architectures (SPAA '20), July 15--17, 2020, Virtual Event, USA}
\acmPrice{15.00}
\acmDOI{10.1145/3350755.3400259}
\acmISBN{978-1-4503-6935-0/20/07}

\usepackage{cleveref}

\newtheorem{observation}{Observation}


\begin{document}

\fancyhead{}

\title{Parallel Planar Subgraph Isomorphism and Vertex Connectivity}

\author{Lukas Gianinazzi}
\affiliation{%
 \department{Department of Computer Science}
  \institution{ETH Zurich}
}
\email{lukas.gianinazzi@inf.ethz.ch}

\author{Torsten Hoefler}
\affiliation{%
 \department{Department of Computer Science}
  \institution{ETH Zurich}
}
\email{htor@inf.ethz.ch}

\renewcommand{\shortauthors}{Gianinazzi and Hoefler}

\begin{abstract}
We present the first parallel fixed-parameter algorithm for subgraph isomorphism in planar graphs, bounded-genus graphs, and, more generally, all minor-closed graphs of locally bounded treewidth. Our randomized low depth algorithm has a near-\emph{linear} work dependency on the size of the target graph. Existing low depth algorithms do not guarantee that the work remains asymptotically the same for any constant-sized pattern.
By using a connection to certain \emph{separating cycles}, our subgraph isomorphism algorithm can decide the vertex connectivity of a planar graph (with high probability) in asymptotically near-linear work and poly-logarithmic depth. Previously, no sub-quadratic work and poly-logarithmic depth bound was known in planar graphs (in particular for distinguishing between four-connected and five-connected planar graphs).
\end{abstract}


\maketitle

\section{Introduction}

Subgraph Isomorphism has applications for pattern discovery in biological networks~\cite{DBLP:journals/jcisd/ArtymiukBGPPRTWW92,Milo824,bioinformatics/bth436}, graph databases~\cite{DBLP:conf/icdm/KuramochiK01}, and electronic circuit design~\cite{DBLP:conf/dac/OhlrichEGS93}. It is also powerful subroutine to solve edge connectivity and vertex connectivity of planar graphs~\cite{DBLP:journals/algorithmica/Eppstein00}. The subgraph isomorphism problem is to look for occurrences of a \emph{pattern graph $H$} as a subgraph of a \emph{target graph $G$}. 
Subgraph isomorphism is a generalization of many $NP$-complete problems (such as finding a \emph{Maximum Clique}, \emph{Longest Path}, or \emph{Hamiltonian Cycle}~\cite{GareyNP-Complete}). The problem remains hard even in bounded degree graphs~\cite{DBLP:journals/tcs/GareyJS76} and planar graphs~\cite{DBLP:journals/ipl/Plesnik79}.

Hence, it is natural to consider \emph{parameterized} versions of the problem that are \emph{tractable} when some parameter is small. We focus our attention to the case when the pattern graph $H$ is relatively small, and give algorithms whose work grows slowly (i.e., close to linear) with the size of the target graph $G$, but is allowed to grow quickly (i.e., exponential) in terms of the size of the pattern graph $H$. This continues the development of \emph{fixed-parameter tractable} (FPT) algorithms for NP-hard problems~\cite{Downey_1999}.

We present a parallel \emph{fixed-parameter tractable} algorithm with low depth for subgraph isomorphism in \emph{planar graphs}. Planar graphs are an important class of graphs which arise naturally from problems in geometry~\cite{De_Loera_2010}, when trying to lay out electronic circuits without crossings~\cite{aggarwal1991multilayer}, and in image segmentation~\cite{DBLP:conf/cvpr/SchmidtTC09}. 

Drawing on existing FPT techniques~\cite{DBLP:conf/cvpr/SchmidtTC09, DBLP:conf/soda/Eppstein95, DBLP:journals/jacm/Baker94}, our algorithm exploits that local neighborhoods of a planar graph are well-behaved and can be efficiently decomposed. We overcome two fundamental challenges:
The first challenge is the reliance on a breadth-first-search (of unbounded depth) to construct the local neighborhoods. We avoid this issue by applying a randomized clustering~\cite{DBLP:conf/spaa/MillerPVX15} into low-diameter parts. This decomposition works because we can bound the probability that an occurrence of the pattern is not in a single cluster by a constant. 
The second challenge is the work-efficient solution of a high depth dynamic program. We transform the problem into a directed acyclic graph and exploit the properties of the parametrized subgraph isomorphism problem to show that introducing shortcuts for only a small subset of nodes suffices to reduce the depth of the graph to poly-logarithmic in the target graph's size (and linear in the pattern graph's size).

\subsection{Preliminaries}

\emph{Subgraph isomorphism} is interested in \emph{occurrences} of a \emph{graph pattern} $H$ (with $k$ vertices and diameter $d$) as a subgraph of a \emph{target graph} $G$ (with $n$ vertices). Formally, a subgraph isomorphism is an injective map $\phi$ from the vertices of $H$ to the vertices of $G$ such that if two vertices $u$ and $v$ are adjacent in $H$, then $\phi(u)$ and $\phi(v)$ are adjacent in $G$. The simplest variant of the subgraph isomorphism problem is to \emph{decide} if any occurrence of the pattern exists in the target graph, but we can also consider \emph{counting} the occurrences or \emph{listing} them.

For any graph $G'$, we denote its vertex set as $V(G')$, its edge set as $E(G')$, and the subgraph of $G'$ induced by a subset $X$ of its vertices by $G'[X]$.
A graph that is formed from the graph $G$ by contracting edges, deleting vertices, and deleting edges is a \emph{minor} of $G$. A family of graphs is \emph{minor-closed} if every minor of every graph in the family is also in the family.

\paragraph{Vertex Connectivity}
A graph with at least $c+1$ vertices is $c$-vertex-connected if removing any $c-1$ vertices does not disconnect the graph. The vertex connectivity $c$ of a graph is the largest number $c$ for which the graph is $c$-vertex-connected.

\paragraph{Tree Decomposition and Treewidth} A \emph{tree decomposition} provides a recursive subdivision of a graph into overlapping subgraphs such that each subgraph is disconnected from the rest of the graph after removing few vertices. The \emph{decomposition tree} records the recursive subdivision in a tree and labels the nodes of the tree with the vertices used to subdivide the graph (in a way that every edge occurs in \emph{at least} one of the tree nodes). See \Cref{fig:tree-decomp} for an example of how a decomposition tree represents a recursive subdivision of a graph. 
\begin{figure}[b]
	\includegraphics[width=1.0\linewidth]{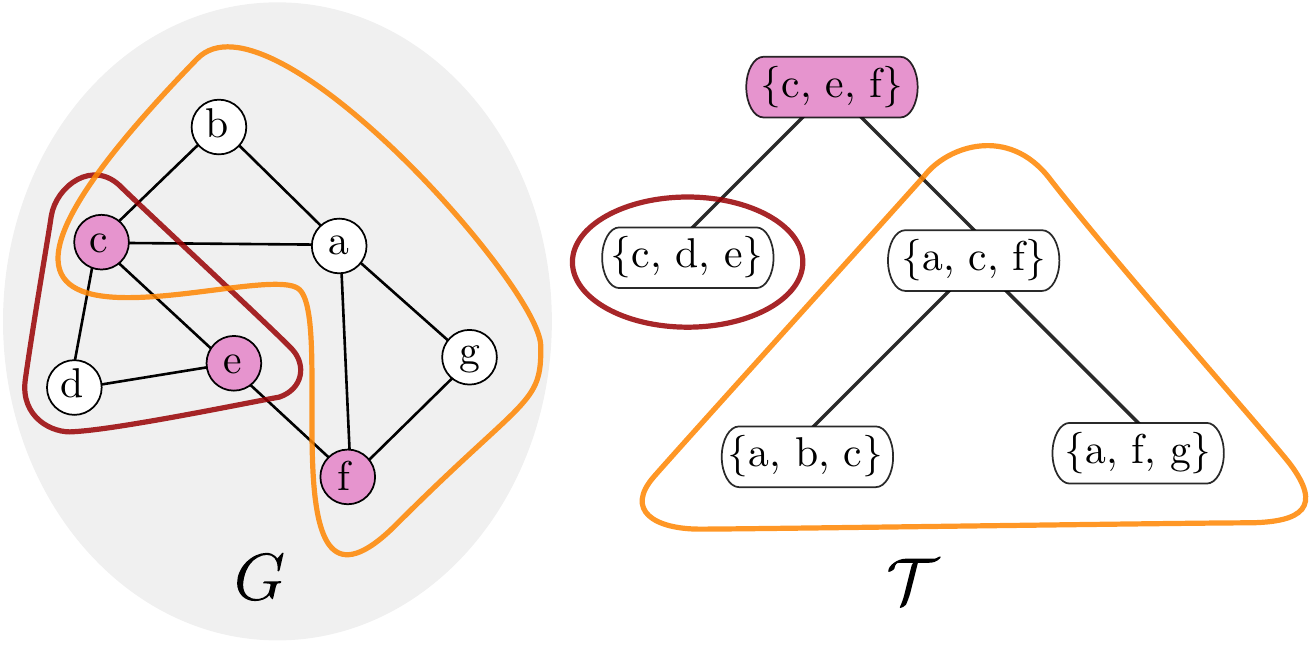}
\caption{Illustration of a graph $G$ and one of its tree decompositions $\mathcal{T}$ of width $2$. The highlighted subtrees in the tree $\mathcal{T}$ correspond to the subgraphs highlighted in the graph $G$ of the same color. The root node $\{c, e, f\}$ separates the two highlighted subtrees, meaning that every path from the subgraph induced by the left subtree to the subgraph induced by the right subtree contains a vertex that is in the root node $\{c, e, f\}$. }
\label{fig:tree-decomp}
\end{figure}

The advantage of the tree decomposition is that it gives a way to describe a divide-and-conquer approach (along some graph decomposition) as a dynamic program on the decomposition tree instead. The dynamic program maintains \emph{partial results} that correspond to the subgraphs of the current node in the decomposition tree and combines the partial results in a bottom-up fashion on this tree.

Formally, a \emph{tree decomposition}~\cite{DBLP:conf/icalp/BodlaenderH95,DBLP:conf/soda/FominLPSW17,DBLP:conf/focs/BodlaenderDDFLP13,DBLP:journals/jal/BodlaenderK96, DBLP:conf/stoc/Bodlaender93, DBLP:journals/jal/Lagergren96} of a graph $G$ consists of a nonempty \emph{decomposition tree} $\mathcal{T}$ where each node $X_i$ of the tree $\mathcal{T}$ is a subset of the vertices $X_i \subseteq V$ of $G$, such that:

\begin{itemize}
	\item Every vertex $u$ of $G$ is contained in a contiguous nonempty subtree of the decomposition tree $\mathcal{T}$.
	\item For every edge $(u,v)$ of the graph $G$, there is a node $X_i$ of the tree $\mathcal{T}$ where both endpoints $u$ and $v$ are in the node $X_i$.
\end{itemize}
The maximum of $|X_i|-1$ over all nodes $X_i$ of the tree $\mathcal{T}$ is the \emph{width} of the tree decomposition. The smallest width of any tree decomposition of $G$ is the \emph{treewidth} $\tau$ of $G$.

We can assume for simplicity that every interior node in the decomposition tree has exactly two children, as we can split high-degree nodes and add empty leaf nodes without changing the width of the decomposition. Moreover, a minimum width tree decomposition of a graph with $m$ edges has $O(m)$ nodes.

\paragraph{Model of Computation} We consider a synchronous \emph{shared memory} parallel machine with concurrent reads and exclusive writes (CREW PRAM). We express our bounds in terms of the total number of operations performed by any execution of the algorithm by all processors (called \emph{work}) and the length of the critical path in the computation (called \emph{depth})~\cite{Blelloch:1996:PPA:227234.227246}. By Brent's scheduling algorithm~\cite{Blelloch:1996:PPA:227234.227246, Reif:1993:SPA:562546}, an algorithm with work $W$ and depth $D$ can be executed with $P$ processors in time $O(W/P+D)$ on a CREW PRAM. 

\paragraph{Randomization} Numerous efficient parallel algorithms make use of some form of randomness~\cite{DBLP:conf/spaa/GeissmannG18}. For some graph problems (such as minimum cuts~\cite{DBLP:conf/spaa/GeissmannG18} and minimum spanning trees~\cite{DBLP:conf/spaa/ColeKT96}), a randomized algorithm has the lowest known bounds. 

We assume each processor has access to an independent and uniformly distributed random word in each time step. If an event occurs with probability at least $1-n^{-a}$ for all constants $a > 1$, we say it occurs \emph{with high probability } (w.h.p.). An algorithm that returns the correct result with high probability is \emph{Monte Carlo}. 

\subsection{Related Work}

For the general case of subgraph isomorphism, no algorithm with less work than the naive $n^k$ is known. Ullmann presents an algorithm that uses a backtracking search~\cite{DBLP:journals/jacm/Ullmann76}.

Tree patterns of bounded size can be found efficiently in general graphs~\cite{DBLP:journals/jacm/AlonYZ95}. Much attention has been put on subgraph isomorphism in special \emph{families of target graphs}, which require some form of sparsity and additional structure~\cite{DBLP:conf/isaac/EppsteinLS10, DBLP:conf/soda/Eppstein95, DBLP:journals/jacm/AlonYZ95, DBLP:journals/siamcomp/ChibaN85}. 

\paragraph{Parameterized Complexity}

The idea behind \emph{parameterized complexity}~\cite{Downey_1999} is to identify (one or more) fundamental parameters $p$ of an NP-hard problem that characterize the difficult part of the problem. Then, a \emph{fixed-parameter tractable algorithm}'s runtime separable into $f(p)g(n)$ (or $f(p)+g(n)$) where $f$ is allowed to be any function of $p$ and $g$ has to be polynomial in $n$~\cite{Downey_1999}.

FPT Algorithms with low depth exist for several NP-complete problems~\cite{DBLP:conf/wg/Bodlaender88,DBLP:conf/wg/ChenH95,DBLP:conf/iwpec/BannachST15}. Refer to \Cref{tab:bounds} (excl. row $1$) for an overview of the FPT algorithms for subgraph isomorphism in planar graphs.

\paragraph{Color Coding}

Using a Monte Carlo technique called \emph{Color Coding}, Alon et al.~\cite{DBLP:journals/jacm/AlonYZ95} obtain $O(e^{k}n^{\tau+1}\log n)$ work on a pattern of treewidth $\tau$, which implies $e^{k}n^{\Theta(\sqrt{k})}\log n$ work for a planar pattern (as the treewidth of a planar graph with $k$ vertices is $\Theta(\sqrt{k})$~\cite{doi:10.1137/0136016,DBLP:journals/jct/RobertsonS86}). The algorithm's depth is poly-logarithmic in $n$ and polynomial in $k$. Their key idea is to color the vertices in the target graph with $k$ random colors, which allows a dynamic programming approach that needs to keep an exponentially smaller state. Note that this algorithm is \emph{not} FPT for the size $k$ of the pattern (nor the treewidth $\tau$), because its runtime grows with $n^{\sqrt{k}}$ (or $n^{\tau+1}$).

\paragraph{Locally Bounded Treewidth}

Eppstein presents the first FPT subgraph isomorphism algorithm for planar graphs that has a \emph{linear} dependency on the size of the pattern graph~\cite{DBLP:conf/soda/Eppstein95}. It runs in polynomial time in $n$ for patterns of size $O(\log n / \log \log n)$. The key insight is to exploit that local neighborhoods of planar graphs have bounded treewidth. 
Their algorithm generalizes to other minor-closed families with a relationship between diameter and treewidth, such as bounded-genus graphs~\cite{DBLP:journals/algorithmica/Eppstein00}. They use a breadth-first-search (BFS) to decompose the graph into these local neighborhoods.

\paragraph{Sampling}
Fomin et al.~\cite{DBLP:conf/focs/FominLMPPS16} present a randomized sampling approach that produces subgraphs of sub-linear treewidth in $k$. Then, they apply an existing FPT dynamic program.

\subsection{Our Contributions}

\renewcommand{\arraystretch}{1.3}
\begin{table}
	\caption{Bounds for deciding planar subgraph isomorphism. ($^{\star}$) The algorithm is Monte Carlo, and its bounds hold w.h.p..}
	\label{tab:bounds}
	\centering
	\small
	\tabcolsep=3mm
	\begin{tabular}{lccc} %
		\toprule
		& Work & Depth \\
		\midrule
		Alon et al.$^{\star}$ \cite{DBLP:journals/jacm/AlonYZ95}  & $e^{k} \ n^{\Theta(\sqrt{k})} \log n$ & $\Theta(k\log n)$ \\
		Eppstein ~\cite{DBLP:conf/soda/Eppstein95} & $O( {2}^{3k \log_2 ( {3 k+1})} n )$  & $\Theta( k n )$ \\
		Dorn ~\cite{DBLP:conf/stacs/Dorn10} & $O(2^{18.81k}n)$ &  $O(2^{18.81k}n)$ \\
		Fomin et al. $^{\star}$ ~\cite{DBLP:conf/focs/FominLMPPS16} & $2^{O(k/\log k)} n^{O(1)}$ & $2^{O(k/\log k)} n^{O(1)}$ \\
		\emph{\textbf{This Paper}}  $^{\star}$ & $O( {2}^{3k \log_2 ({3 k+1})} \ n \log n)$ & $O(k \log^2 n)$  \\
	\end{tabular}
\end{table}

We present the first FPT work planar subgraph isomorphism algorithm with depth poly-logarithmic in $n$ and polynomial in $k$. Our Monte Carlo algorithm has $k^{O(k)} n \log n$ work and $O(k\log^2 n)$ depth in planar graphs and has FPT work in all minor-closed families of graphs of locally bounded treewidth (see \Cref{sec:apex-free}).

\Cref{tab:bounds} contains the exact bounds and a comparison to the related works regarding planar graphs. Note that if the pattern graph occurs in the target graph, the expected work is $k^{O(k)} n$. Our algorithm can also \emph{list} all $x$ occurrences of a pattern with  $O(xk \ (\log n + \log x)) + k^{O(k)} n \ (\log n + \log x) \  $ work and $O(k \log ^2 n \ (\log n + \log x))$ depth.

We use a \emph{low-diameter decomposition}, which can ensure that the occurrences of the pattern graph are in the same low-diameter part of the graph with sufficient probability. Then, we show how to exploit the special structure of a tree decomposition based algorithm to compute its results work-efficiently in parallel. Finally, we provide a randomized extension to the algorithm that also handles disconnected pattern graphs.

More generally, we can find isomorphic subgraphs that separate a set of marked vertices (leaving them in different components after removal of the subgraph).
Because there is a relation between finding certain separating cycles as subgraphs and planar vertex connectivity, our subgraph isomorphism algorithm yields better parallel bounds for deciding vertex connectivity in planar graphs.

We show that planar vertex connectivity can be answered in $O(n \log n)$ work and $O(\log ^ 2 n)$ depth. Previously, only $2$-connectivity and $3$-connectivity had sub-quadratic work and poly-logarithmic depth solutions~\cite{DBLP:journals/siamcomp/TarjanV85, DBLP:conf/stoc/MillerR87}.

\section{From Planar To Low Treewidth}

Planar graphs do not have bounded treewidth (it can be up to $\sqrt{n}$), which prevents a direct application of bounded treewidth techniques (as we use in \Cref{sec:bounded-treewidth-algo}). Fortunately, a planar graph of diameter $d$ has treewidth at most $3d$~\cite{DBLP:conf/soda/Eppstein95}, and each occurrence of a pattern with diameter $d$ is contained in a subgraph of diameter $d$ of the target graph. 

Hence, a simple (but work-inefficient) approach to solve subgraph isomorphism in planar graphs would consist of building for every vertex in the target graph the subgraph induced by nodes at a distance at most $d$, and then invoking an algorithm for bounded treewidth graphs on each of those subgraphs. This approach of \emph{covering} the graph is inefficient because many vertices of the target graph could be in multiple (even all) of these subgraphs, leading to a total size of these subgraphs of $\Theta(n^2)$.

Instead, Eppstein~\cite{DBLP:conf/soda/Eppstein95} proposed (based on an idea by Baker~\cite{DBLP:journals/jacm/Baker94}) a covering approach based on a single BFS to cover all subgraphs of diameter at most $d$ with graphs of total size only $O(dn)$. It is easy to see that naive BFS takes linear work and $O(D)$ depth on a diameter $D$ graph, but we care exactly about the situation when the diameter $D$ is not bounded. Even on planar graphs, performing work-efficient and low-depth BFS is a challenging problem. An approach by Klein~\cite{DBLP:conf/focs/KleinS93} achieves $O(n \log^9 n)$ work and poly-logarithmic depth.

\begin{figure}
\vspace{1em}
	\includegraphics[width=1.0\linewidth]{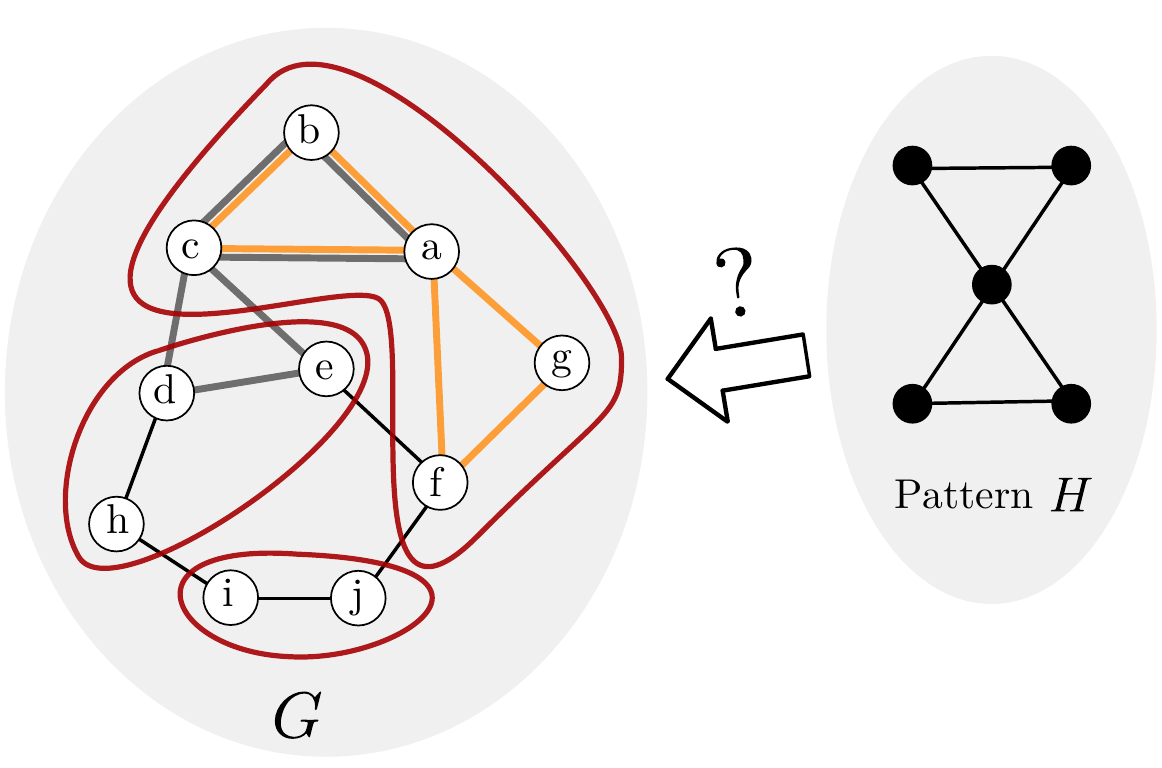}
\caption{A randomized procedure splits the target graph $G$ into clusters. Each occurrence of the pattern $H$ is contained inside a single cluster with constant probability. In the example, the occurrence of the pattern $H$ with vertices $f, g, a, b, c$ is contained in a single cluster, but the occurrence $d, e, a, b, c$ crosses the clusters. Hence, the former is found with this clustering, but the latter is not.}

\label[figure]{fig:clustering}
\end{figure}

To avoid the issue of low-depth BFS, in our approach, we first decompose the graph into \emph{randomized} clusters of small diameter (as illustrated in \Cref{fig:clustering} and \Cref{fig:k-cover}). This allows us to then run a simple parallel BFS on those low diameter graphs and construct a covering for each of those clusters. In summary, one run of our subgraph isomorphism algorithm works as follows:
\begin{enumerate}
	\item Cover the target graph with subgraphs $G_0, \dotsc, G_i$ of bounded treewidth (they might overlap, as detailed in \Cref{sec:low-treewidth-cover}).
	\item Solve subgraph isomorphism for each such bounded treewidth subgraph in parallel (as described in \Cref{sec:bounded-treewidth-algo}).
\end{enumerate}
Since our covering algorithm is randomized, an occurrence of a pattern may not be contained in any single subgraph in the cover. However, in expectation $O(1)$ repetitions suffice to find an occurrence of the pattern if it exists. At most $O(\log n)$ runs suffice to certify that no occurrence of a pattern exists with high probability.
Our main result for planar graphs is the following:

\begin{theorem}
Deciding (with high probability) if a connected pattern graph $H$ occurs as a subgraph of a planar target graph $G$ takes $O((3k)^{3k+1} n \log n)$ work and $O(k \log ^2 n)$ depth.
\end{theorem}

\noindent
For a pattern of small diameter $d$, we obtain better bounds:
\begin{corollary}
Deciding (with high probability) if a connected pattern graph $H$ of diameter $d$ occurs as a subgraph of a planar graph $G$ takes $O((3d+3)^{3k+1} n \log n)$ work and $O(k \log ^2 n)$ depth.
\end{corollary}

To simplify the exposition, we assume (for now) that the pattern graph is connected and focus on the decision version of the problem. We then discuss how to remove the assumption of connectedness in \Cref{sec:disconnected} and show how to modify the algorithm to list all occurrences of a pattern graph in \Cref{sec:listing-si}. Moreover, we generalize the approach from planar graphs to a class of graphs that contains all bounded-genus graphs in \Cref{sec:apex-free}.

\subsection{Parallel Low-Treewidth Cover}\label{sec:low-treewidth-cover}

We show how to construct (in parallel) a set of subgraphs of low treewidth such that each occurrence of a connected pattern $H$ is in at least one of the subgraphs with constant probability. The first step is to use a \emph{low-diameter decomposition}. 
The goal of a low-diameter decomposition is to partition the vertices of the graph into (vertex-disjoint) clusters of low diameter such that few edges of the graph connect vertices in different clusters.

\emph{Exponential Start Time Clustering}~\cite{DBLP:conf/spaa/MillerPVX15} is especially well-suited for our purposes because it bounds the \emph{probability} that an edge connects two different clusters. This observation allows us to bound the probability that a connected subgraph is split into multiple clusters, and thus the clustering preserves the occurrences of a graph pattern with nontrivial probability, as needed for our purposes. 

A \emph{clustering} of $G$ is a set of vertex-disjoint induced subgraphs called \emph{clusters} that together contain all vertices. We say an edge \emph{crosses the clusters} if it has endpoints in the vertex sets of two distinct clusters.

\begin{lemma}[Exponential Start Time Clustering~\cite{DBLP:conf/spaa/MillerPVX15}]\label[lemma]{lemma:exp-start}
With $O(n)$ work and $O(\beta \log n)$ depth, Exponential Start Time $\beta$-Clustering produces, w.h.p., clusters of diameter $O(\beta \log n)$ where each edge crosses the clusters with probability at most $1/\beta$. 
\end{lemma}
Note that Exponential Start-Time Clustering does not allow us to fix the number of clusters a priori. Instead, the number of clusters depends on the structure of the graph. For example, a clique will most likely end up as a single low-diameter cluster.

Because every edge crosses the clusters with small probability, the probability that a fixed occurrence of the pattern contains an edge that crosses the clusters is also relatively small (for an appropriate choice of parameter $\beta$). See \Cref{fig:clustering} for an illustration.

\begin{observation}\label{obs:clustering}
	The probability that no edge of a connected subgraph $H$ of the graph $G$ crosses a cluster of an Exponential Start Time $2k$-Clustering of $G$ is at least $1/2$
\end{observation}
\begin{proof}
The idea is that some spanning tree of the occurrence remains intact (i.e., no edge in the tree crosses a cluster) with the given probability, which implies the result.
	Consider an arbitrary spanning tree $A$ of $H$. By \Cref{lemma:exp-start}, the probability that a particular edge of the spanning tree crosses the clusters is at most $\frac{1}{2k}$. By the union bound, the probability that any of the $k-1$ edges of the spanning tree $A$ crosses the clusters is at most $\frac{k-1}{2k} < \frac{1}{2}$. Hence, the probability that no edge crosses the clusters is at least $1/2$.
\end{proof}

\noindent
We combine the clustering idea with the approach from Eppstein~\cite{DBLP:conf/soda/Eppstein95} and Baker~\cite{DBLP:journals/jacm/Baker94} for the \emph{Parallel treewidth $k$-$d$ cover} algorithm.
\begin{figure}[b]

	\includegraphics[width=1.0\linewidth]{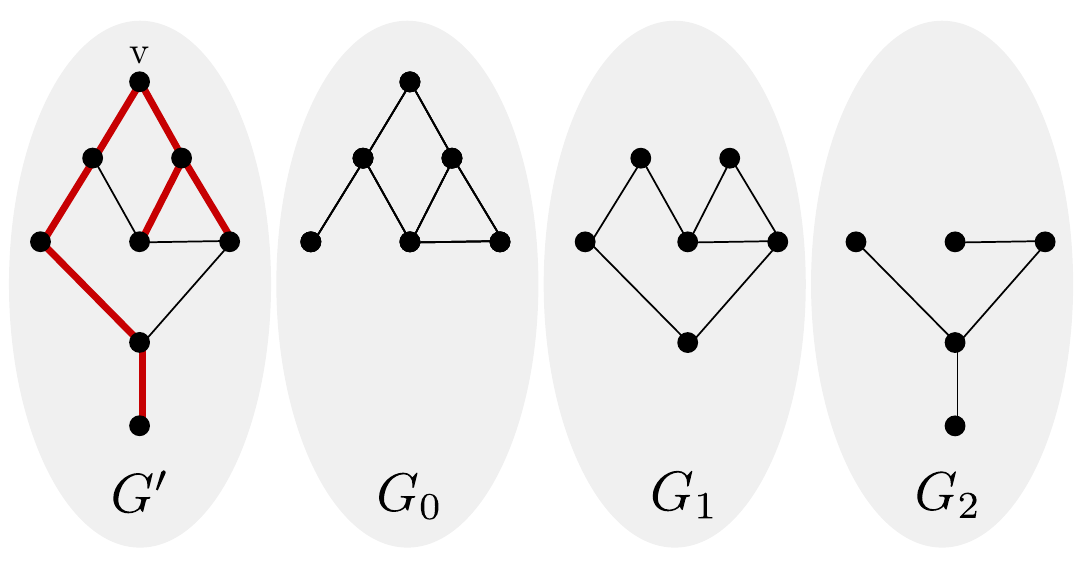}
\caption{A cluster $G'$ is covered by a set of subgraphs that are each induced by $d$ consecutive levels in a BFS tree of the cluster. In the example, $d=2$. The bold BFS tree of the graph $G'$ rooted at $v$ guides the covering of the graph with the induced subgraphs $G_0, G_1, G_2$. If the original graph contains a pattern of diameter $d$, then at least one of the subgraphs does as well. Note that the graphs $G_3$ and $G_4$ are not needed because their diameter is smaller than the pattern's diameter.}
\label[figure]{fig:k-cover}
\end{figure}

\vspace{1em}
\textbf{Parallel Treewidth $k$-$d$-Cover.}
\begin{enumerate}
	\item Run Exponential Start Time $2k$-Clustering on $G$.
	\item For each cluster, choose an arbitrary root $v$ and run a naive parallel BFS within the cluster.
	\item This yields a BFS tree for each cluster. For each level $i$ of the tree, output the subgraph $G_i$ induced by the vertices at distance $i$ through $i+d$ from $v$ (as illustrated in \Cref{fig:k-cover}).
\end{enumerate}
\vspace{1em}

\pagebreak

The algorithm guarantees that each of the subgraphs has low treewidth and that every occurrence of the pattern graph is in at least one of the subgraphs with constant probability: 
\begin{theorem}
For a planar target graph $G$ and a connected pattern graph $H$ with $k$ vertices and diameter $d$, a Parallel Treewidth $k$-$d$ Cover produces
 a set of induced subgraphs $G_i$ of $G$ such that:
\begin{itemize}
	\item Every graph $G_i$ has treewidth at most $3d$.
	\item Every vertex of $G$ is contained in at most $d$ graphs $G_i$.
	\item Every fixed occurrence of $H$ is contained in at least one of the graphs $G_i$ with probability at least $1/2$.
\end{itemize}
The algorithm takes, w.h.p., $O(n d)$ work and $O(k\log n)$ depth.
\end{theorem}
\begin{proof}
Each of the graphs $G_i$ is a subgraph of a planar graph with diameter $d$. Hence, it has treewidth at most $3d$~\cite{DBLP:conf/soda/Eppstein95}.
By \Cref{obs:clustering}, an occurrence $H'$ of $H$ is in the same cluster with probability at least $1/2$. 
 If this is the case, consider the first vertex $u$ of the pattern occurrence $H'$ encountered during the BFS done for the cluster and let $i$ be the distance of $u$ from the root $v$ of the BFS tree. Then, the occurrence $H'$ is an induced subgraph of $G_i$.
 
The clusters have diameter $O(k \log n)$. Hence, the BFSes have $O(k \log n)$ depth. Each vertex and edge is part of at most $d$ subgraphs by construction, which implies that the work is $O(nd)$. 
\end{proof}

It remains to find (in parallel) occurrences of the pattern on each of the low treewidth subgraphs we constructed. The algorithm in \Cref{sec:bounded-treewidth-algo} requires that a tree decomposition of the subgraph has already been computed. For a planar graph, constructing such a decomposition of width $3d$ takes $O(n)$ work and $O(d)$ depth given a planar embedding of the graph~\cite{DBLP:conf/soda/Eppstein95, DBLP:journals/jacm/Baker94}. Computing a planar embedding takes $O(n)$ work and $O(\log ^2 n)$ depth~\cite{DBLP:conf/focs/KleinR86}.


\pagebreak

\section{Algorithm for Bounded Treewidth}\label{sec:bounded-treewidth-algo}

The main result of this section is a parallel algorithm to solve subgraph isomorphism in parallel on graphs of bounded treewidth. It is based on a simplified version of the algorithm from Eppstein~\cite{DBLP:conf/soda/Eppstein95}. We transform the original problem into a graph search problem. Exploiting the particular structure of the resulting acyclic graph allows us a low depth and work-efficient solution.

\begin{lemma}
Deciding if a connected pattern graph $H$ is isomorphic to a subgraph of the target graph $G$ of treewidth $\tau$ takes $O(k\log^2 n)$ depth and $O( (\tau+3)^{3k+1} n)$ work. The bounds hold w.h.p..
\end{lemma}

The overall idea of the sequential algorithm is to gradually compute the subgraph isomorphism while traversing the decomposition tree in a bottom-up fashion. We start by discussing the \emph{partial matches} (partially completed subgraph isomorphisms) the algorithm employs, which are crucial for the parallel algorithm as well.

\subsection{Partial Matches}

Every node $X$ in the decomposition tree corresponds to a subgraph $G[X]$ induced by $X$ in the target graph $G$ with only a small number of vertices $\tau+1$. Moreover, the descendants of the node $X$ (together with $X$) induce a subgraph $G_X$ of the graph $G$ that is separated from the rest of the graph $G$ by the vertices in the tree decomposition node $X$. The idea of \emph{partial matches} is to find occurrences of sub-patterns of the pattern $H$ within these subgraphs and combining them in a bottom-up fashion in the tree decomposition.

\emph{Partial matches} exist between subgraphs of the pattern graph $H$ and these induced subgraphs $G_X$. Because vertices that are in the subgraph $G_X$ but are not in the separating set $X$ are not directly connected to the rest of the graph $G$, it is not necessary to explicitly store the mapping between pattern and target graph for these vertices in order to combine a partial match inside this subgraph with partial matches from the rest of the graph. Hence, when we build partial matches, only the $\tau^k$ different mappings for these vertices in the separating set $X$ are important. The remaining vertices that have already been matched in a child are recorded as such. See \Cref{fig:partial-match} for an example.

\begin{figure}
		\includegraphics[width=1.0\linewidth]{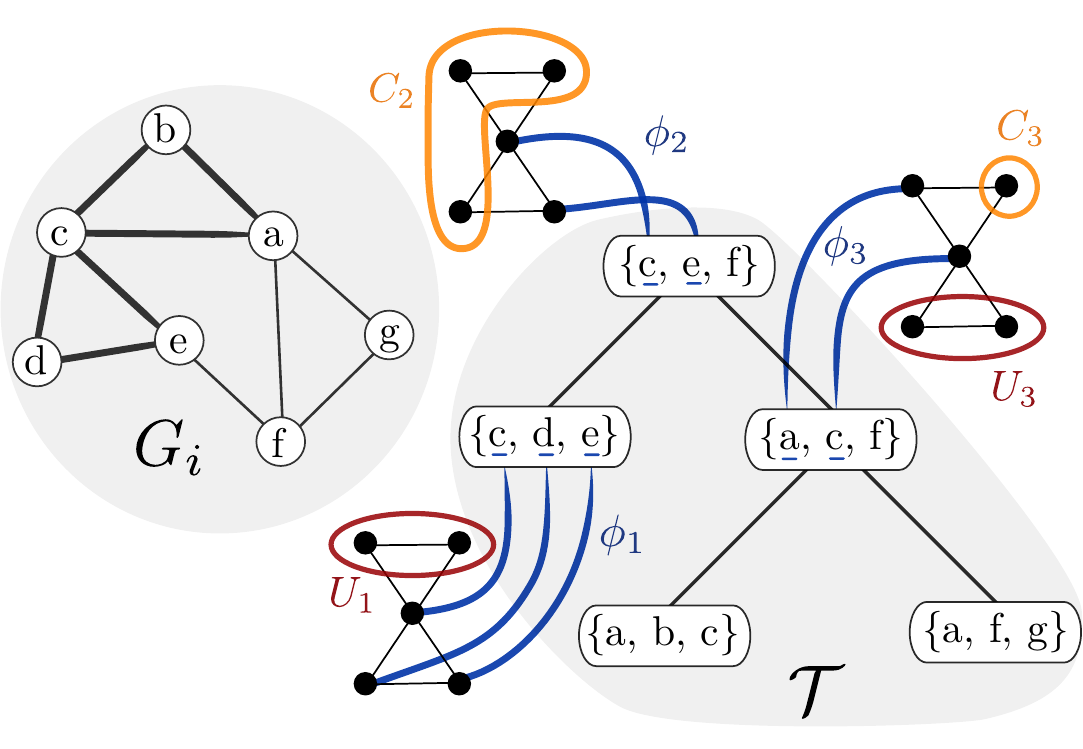}
			\vspace{-2em}
	\caption{A valid partial match of the root $\{c, e, f\}$ of the decomposition tree is built from two compatible (and valid) partial matches of the left child $\{c, d, e\}$ and the right child $\{a, c, f\}$. Observe how every vertex in the pattern that is marked as `\emph{matched in a child}' at the root is matched in exactly one of the two other partial matches. Also, note how the partial matches agree on the vertices that the decomposition nodes have in common: the partial match of the root agrees with the left child's match on $c$ and $e$ and with the right child's match on $c$.
		\vspace{-1em}
}

\label{fig:partial-match}
\end{figure}

Formally, a partial match of $X$ is a triple $(\phi, C, U)$, where $C$ denotes the set of vertices \emph{matched in a child}, $U$ the set of vertices marked as \emph{unmatched}, and a subgraph isomorphism function $\phi$ from the subgraph $H[V(H)\backslash (C \cup U)]$ to the subgraph $G[X]$. 
If a vertex $v$ of $H$ is \emph{matched in a child}, the vertex $v$ is mapped to a vertex in $G_X$ which does not appear in $X$.
If a vertex $v$ of $H$ is \emph{unmatched}, then it is not matched to any vertex that appears in the subgraph $G_X$. 

\subsection{Eppstein's Sequential Algorithm} The idea is to extend the \emph{partial matches} while traversing the decomposition tree $\mathcal{T}$ bottom-up. 
The goal is to construct a partial match of the root node where no vertex is \emph{unmatched}. We focus on how to construct such a partial match for the root, from which a specific subgraph isomorphism can be recovered efficiently (by collecting appropriate partial isomorphisms in a top-down traversal of the tree; see also ~\Cref{sec:listing-recovering}). 

A partial match of a child node $Y$ can be extended when going to a parent node $X$ by matching some additional vertices that were unmatched by the child match to $X$, marking the vertices that have been matched by the child but are not in the parent as \emph{matched in a child}, and leaving the rest of the partial isomorphism function the same (the vertices that were not newly matched in $X$ remain unmatched).

A partial match that can be extended to a parent's partial match (possibly together with another child's partial match) is called \emph{consistent} with a parent's partial match.
The precise rules for being \emph{consistent} follow. Consider a node $X$ of the decomposition tree, one of its children $Y$, and the partial matches $(\phi_X, C_X, U_X)$ of $X$ and $(\phi_Y, C_Y, U_Y)$ of $Y$.
For all vertices $v$ in $H$:
\begin{itemize}
	\item If $v$ is matched by $\phi_Y$ to a node in $X$ \emph{or} by $\phi_X$ to a node in $Y$, then they map to the same value: $\phi_X(v) = \phi_Y(v)$. This prevents the partial matches $(\phi_X, C_X, U_X)$ and $(\phi_Y, C_Y, U_Y)$ to map the same vertex in the pattern graph to different nodes in the target graph. 
	\item If the child partial match $\phi_Y$ matches a vertex $v$ to a vertex not in the parent label set $X$ \emph{or} marks the vertex $v$ as \emph{matched in a child} (i.e., in $C_X)$, then the parent partial match marks the vertex $v$ as \emph{matched in a child} (i.e., in $C_Y$). In particular, we have $C_Y \subseteq C_X$.  
\end{itemize}
Note that these rules imply that the child's partial match does not match any vertex that is unmatched by the parent, i.e.,  $U_Y\subseteq U_X$. 

The point of the following combination rule is to ensure (on top of consistency) that a vertex that is marked as \emph{matched in a child} in the parent is matched in exactly one of the children. 
A partial match $M_X$ of node $X$ is \emph{compatible} with a partial match $M_L$ of the left child $L$ of $X$ and partial match $M_R$ of the right child $R$ of $N$ if the following conditions hold:
\begin{itemize}
	\item The partial matches $M_L$ and $M_R$ are both consistent with the partial match $M_X$.
	\item If a vertex is marked as \emph{matched in a child} by $M_X$, then it is marked as \emph{unmatched} in exactly one of the child matches $M_L$ and $M_R$. 
\end{itemize}

A partial match is \emph{valid}, if it is compatible with two partial matches of its children, or if it does not mark any vertices as \emph{matched in a child}.
Note that the trivial partial match that marks everything as \emph{unmatched} is always valid.
A valid partial match of the root node that does not mark any vertex as unmatched certifies the existence of a subgraph isomorphism.

The sequential algorithm traverses the decomposition tree bottom-up and enumerates all possible partial matches for the current node, then checks which are valid (given the valid matches for the children). For a tree decomposition of width $\tau$ and a pattern of size $k$, there are at most $(\tau+3)^k$ possible partial matches per node. There are at most $(\tau+3)^{3k}$ combinations of partial matches of the parent and its two children and validating a combination takes $O(\tau)$ time. Hence, the overall runtime is $O((\tau+3)^{3k+1}m)$.

\subsection{Parallel Algorithm}

The issue is that even a low-diameter planar graph might have a decomposition tree that has a large height of $\Omega(n)$. Therefore, parallelizing the computation at each node of the decomposition tree is not enough. It is possible to transform any tree decomposition into a decomposition of height $O(\log n)$ with \emph{three times} the treewidth~\cite{DBLP:conf/icalp/BodlaenderH95}, which increases the work by a factor of $\Omega(9^{k})$.

To avoid this, we parallelize across the height of the decomposition tree.
In order to obtain a simpler problem, we partition the tree into paths. Then, we solve the problem on each of the paths. A path can be solved once all paths that start at a child of a node in the path have been solved. We avoid the sequential bottleneck by transforming the problem of finding valid partial matches in these subpaths of the tree decomposition into a reachability question in an acyclic directed graph with special structure. The reachability question can be solved work-efficiently with a low depth on this acyclic graph by introducing shortcuts of exponentially increasing distance to a carefully selected subset of the vertices.
 
\subsubsection{Decomposition into Paths}

Let us start by discussing how to decompose the tree into suitable subpaths. Walk from every leaf towards the root until reaching a branching node (i.e., a node with at least $2$ children). Remove the visited paths from the tree, and proceed recursively. This decomposition can be implemented efficiently using parallel expression tree evaluation (tree contraction)~\cite{Reif:1993:SPA:562546,DBLP:conf/focs/MillerR85}:

\begin{lemma}[\Cref{sec:bough-decomp}] ~\label{lem:bough-decomposition}
	A tree $T$ can be decomposed into a set of paths $P$ where the paths are grouped into $O(\log n)$ layers with the property that vertices in the $i$-th layer have no children in a layer larger than $i$. This decomposition takes $O(n)$ work and $O(\log n)$ depth.
\end{lemma}

\subsubsection{The Graph of Partial Matches.}

We can reason about how to construct the valid partial matches for a subpath $\mathcal{P}$ of the tree decomposition, assuming we already solved all paths descending from a child of $\mathcal{P}$. Specifically, we derive locally at each node in the subpath $\mathcal{P}$ a set of partial matches that are valid partial matches \emph{if at least} one of the partial matches of a child node of $\mathcal{P}$ is also a valid partial match. At the leaf of the path, we know which partial matches are valid (because both children have already been solved). This observation leads to the idea to construct a directed acyclic graph of partial matches where reachability models the validity of the partial matches, as follows.

Let $\mathcal{P}$ be a subpath of the tree decomposition $\mathcal{T}$. Consider a node $X$ in the path $\mathcal{P}$ and assume we already computed the partial matches for the left child $L$ of $X$ (the other child is the right child $R$, where $R\in \mathcal{P}$). Then, we can check which partial matches of $X$ and the right child of $X$ are compatible with a partial match of $L$. This yields for every partial match $M_X$ of $X$ a set of partial matches of $R$ that would validate the partial match $M_X$.

We construct a directed acyclic graph $G'$ based on this idea. For the leaf node of $\mathcal{P}$, there is a vertex in $G'$ for every valid partial match. For every other node $X$ in $\mathcal{P}$, there is a vertex for every partial match of that node $X$. Then, there is an edge from a partial match $M_R$ of the child $R$ of $X$ to a partial match $M_X$ if there is a valid partial match $M_L$ of the other child $L$ of $X$ such that $M_X$ is compatible with $M_L$ and $M_R$.

Reachability in the graph $G'$ can model which partial matches are valid: A partial match is tagged as valid if it does not mark any vertices as \emph{mapped by a child}. The partial matches from the leaf node of $\mathcal{P}$ are also tagged as valid.
Then, the valid partial matches are those that are reachable from a partial match tagged as valid in the directed acyclic graph $G'$.

\subsubsection{Finding Valid Partial Matches Via Reachability}

Next, we discuss how to compute all the valid partial matches using the directed acyclic graph $G'$. Note that this graph $G'$ still has a diameter equal to the length of the path $\mathcal{P}$, so we cannot directly use BFS. Hence, we introduce \emph{shortcuts} of exponentially increasing distance to reduce the diameter to $O(k\log n)$. After introducing the shortcuts, we use naive parallel BFS to determine all the reachable vertices. The details follow. 

A simple (but a factor $\log n$ work-inefficient) way to solve reachability is to introduce shortcuts for every vertex (similarly to some list ranking and connected components algorithms~\cite{Reif:1993:SPA:562546}):

\begin{enumerate}
	\item Introduce shortcuts in $\log n$ rounds $0, 1, \dotsc, \log n$. 
	\item Round $i$ creates shortcuts of length $2^{i}$. The edges of the graph are shortcuts of length $1$.
	\item For round $i>0$, for every vertex $u$, look at all its outgoing edges of length $2^{i-1}$. For each such edge $(u,v)$, look at all edges $(v, w)$ of equal length $2^{i-1}$ and add an edge $(u,w)$ of length $2^{i}$ to $u$.
\end{enumerate}
This would result in $O(\log n)$ depth, but also be work inefficient by up to a factor $\Theta(\log n)$ (when $k=O(1)$) because \emph{every} vertex in the graph $G'$ does $\Omega(\log n)$ work.

The crucial observation to overcome this limitation is that any valid partial match is constructed by matching a ``new'' vertex at most $k$ times. Thus, there are at most $k$ edges in $G'$ that match new vertices along any path in $G'$ towards a valid partial match. The rest of the edges in $G'$ do not introduce any new matches, but instead, translate from the partial match of a child to an equivalent partial match of the root. Since there is only one way not to introduce any new matches  (see \Cref{fig:trivial-match}), the subgraph of edges that do not introduce new edges is a directed forest (where edges are directed towards the roots). Hence, it suffices to introduce shortcuts in this forest $F$.

\begin{figure}
		\includegraphics[width=1.0\linewidth]{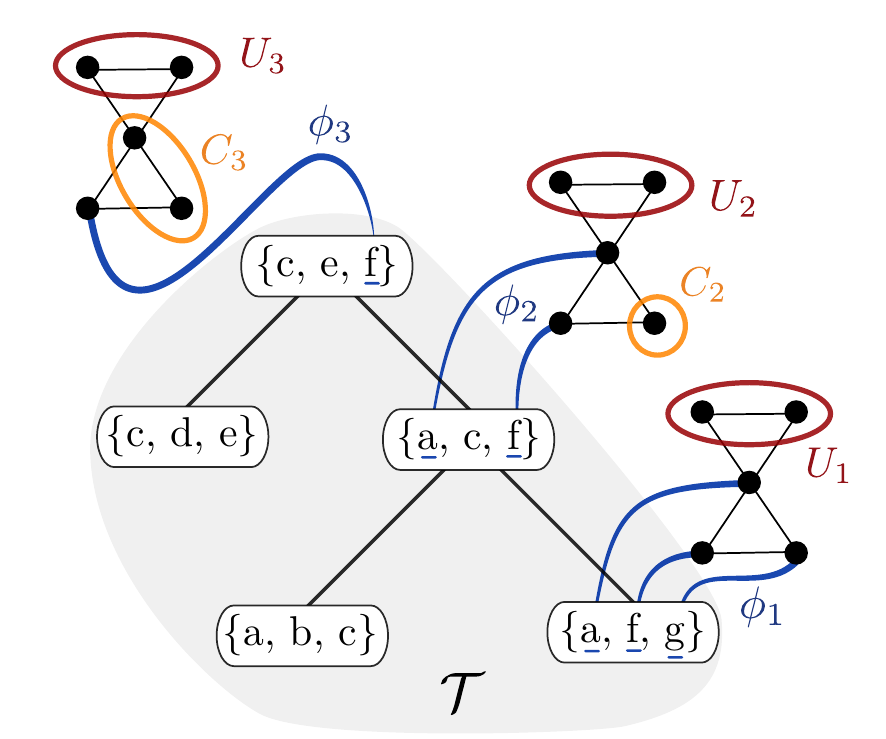}
		\vspace{-2em}
	\caption{
	The valid partial match $(\phi_1, C_1, U_1)$ at node $\{a,f,g\}$ in the decomposition tree $\mathcal{T}$ can be turned into a valid partial match $(\phi_2, C_2, U_2)$ of the parent $\{a, c, f\}$ without matching a new vertex in exactly one way: The partial match has the same set of unmatched vertices $U_2=U_1$. The set of children vertices contains the vertex that was matched to $g$, because $g$ is not in $\{a,c,f\}$. The isomorphism function $\phi_2$ is the same as $\phi_1$ on all the vertices in $\{a,f,g\} \cap \{a ,c, f\}$ and undefined elsewhere. The partial match $(\phi_2, C_2, U_2)$ is turned into $(\phi_3, C_3, U_3)$ similarly, except that now the set $C_3$ of vertices \emph{matched in a child} also includes those in $C_2$}
	\label{fig:trivial-match}
	\vspace{-1em}
\end{figure}

Because the subgraph $F$ is a forest, shortcuts can be introduced work-efficiently in parallel:
In each tree of $F$, decompose the tree into paths using \Cref{lem:bough-decomposition}. In each path, choose every $\log n$-\emph{th} vertex as a vertex where shortcuts are introduced. Add a shortcut from every such vertex to the next, then add shortcuts of exponentially increasing distance between them (within the path). Moreover, add a shortcut from every vertex to the first vertex in a lower layer.
\begin{lemma}
	Computing the valid partial matches of the graph pattern $H$ in a subpath $\mathcal{P}$ of a decomposition tree $\mathcal{T}$ of width $\tau$ takes $O(|\mathcal{P}| ((\tau+3)^{3k+1}))$ work and $O(k\log n)$ depth.
\end{lemma}
\begin{proof}

The work is linear in the number of vertices because we add the edges of exponentially increasing distances to a forest $F$ of $O(|\mathcal{P}|/\log n)$ vertices.

After introducing the shortcuts, the distance from a valid leaf node to any other valid node is $O(k\log n)$: Consider any path $p$ in the original graph $G'$. It contains at most $k$ edges that are not in the forest $F$. 
Therefore, it consists of at most $k$ subpaths $p_1, \dotsc , p_k$ where each $p_i$ is a subgraph of the forest $F$. Each subpath $p_i$ is contained in a maximal tree $F_i$ of $F$. 
By \Cref{lem:bough-decomposition}, $p_i$ intersects at most $O(\log n)$ subpaths of $F_i$. It takes $O(\log n)$ hops to move from the first such subpath to the last (because of the shortcuts to a vertex in a lower layer). Then, it takes an additional $O(\log n)$ hops to traverse the first and last subpath using the shortcuts within each subpath.
We conclude that the overall number of hops to traverse the path $p$ is $O(k \log n)$. 

Together with the depth of constructing the shortcut graph, this means that the depth of the algorithm is $O(k \log n)$.
\end{proof}

\section{Extensions}

We generalize our algorithm to disconnected patterns, show how to list all occurrences of a graph pattern, and characterize the family of graphs for which the algorithm is still FPT.

\subsection{Disconnected Patterns} \label{sec:disconnected}

We extend our algorithm so that it can handle arbitrary disconnected patterns. These patterns are challenging because (in particular) the algorithm for treewidth $k$-cover cannot guarantee that every component of the pattern graph is in the same cluster. 

Consider a pattern graph $H$ consisting of $l$ connected components. Number the components arbitrarily from $1$ to $l$. A naive approach is to try out all $l^n$ possible ways to split the target graph into $l$ components.
A randomized approach (inspired by \emph{color coding}~\cite{DBLP:journals/jacm/AlonYZ95}) allows us to remove the exponential dependency on the number of vertices $n$. It works as follows:

\begin{enumerate}
	\item Color each vertex in $G$ independently and uniformly at random with a number between $1$ and $c$.
	\item For each color $i$, let $G^{i}$ be the subgraph induced by the vertices that have color $i$.
	\item Search for occurrences of the $i$-th component of $H$ in the subgraph $G^{i}$ of color $i$ vertices.
	\item Return true if and only if each search is successful.
\end{enumerate}

\begin{lemma}
Finding (with high probability) an occurrence of a disconnected pattern with $l$ components and $k$ vertices takes $O(l^{k}\log n)$ more work than finding an occurrence of a connected pattern. 
\end{lemma}
\begin{proof}
Consider a fixed occurrence of the pattern $H$. The probability that all of its vertices are assigned to the correct component of $H$ is $l^{-k}$. Hence, $O(l^k)$ repetitions suffice to find a particular occurrence of $H$ with constant probability, and $O(l^k \log n)$ repetitions suffice to certify that no occurrence exists with high probability.
\end{proof}
Note that this technique of finding disconnected patterns by reduction to the connected case is completely general and can be used in conjunction with any subgraph isomorphism algorithm.

\subsection{Listing all Occurrences}\label{sec:listing-si}

We describe the modifications necessary to make our algorithm list all occurrences of a pattern. The first step is to modify the algorithm such that it returns a particular occurrence of a pattern with probability at least $1/2$. Then, we can repeatedly generate a new set of occurrences, remove duplicates (by hashing), until we are confident enough that we have found all occurrences. The main difficulty is that the number of iterations necessary to find all the occurrences depends on the number of occurrences, which we do not know in advance.

However, since every \emph{particular} occurrence is found with probability at least $1/2$ in each iteration, if there is an occurrence that has not yet been found, at least one new occurrence is found with probability at least $1/2$. This argument shows that the process is related to getting many heads in a row when flipping coins: it is unlikely that many iterations in a row do not find a new occurrence.

\begin{observation}\label{obs:coins}
	For all $j\leq i$, the probability that in a sequence of $j$ independent coin flips $i$ heads occur in a row is at most $j 2^{-i}$.
\end{observation}
\begin{proof}
	The probability that $i$ heads occur in a row starting from the $y$-th coin flip is at most $2^{-i}$. By a union bound over the $j$ possible start positions, the bound follows.
\end{proof}
This observation still holds even for biased coins, as long as the probability that heads comes up is \emph{at most $1/2$}.

Therefore, we iterate until after $j$ iterations we have seen no new occurrence for $\log_2 j + \Theta(\log n)$ iterations in a row to guarantee that we have found all occurrences with high probability in $n$.

\begin{theorem}
Listing w.h.p. all $x$ occurrences of a connected pattern graph in a planar target graph takes $O(k \log ^ 2 n \ (\log (x) + \log n)))$ depth and $O( (xk + (3k+3)^{3k+1} n) \ (\log n + \log x))$ work. 
\end{theorem}
\begin{proof}
Every iteration finds a specific occurrence with probability at least $1/2$. Hence, after $\log_2 x+\Theta(\log n)$ iterations, the probability that we have \emph{not} found a specific occurrence is at most $x^{-1}n^{-\Omega(1)}$. By a union bound over the $x$ occurrences, the probability that we have not found \emph{all} occurrences is at most $n^{-\Omega(1)}$.
Hence, after $i=\log_2 x+\Theta(\log n)$ iterations, the algorithm will, with high probability, not find any new occurrences (because there are none) and by construction terminate after an additional $O(\log i + \log n)$ iterations.
Overall, the algorithm takes at most $O(\log x+\log n)$ iterations to terminate with high probability. Together with the bounds from \Cref{sec:listing-recovering} this implies the work and depth bounds.

We show that the probability that the algorithm terminates before all occurrences have been found is at most $n^{-\Omega(1)}$. Consider the longest prefix of iterations of the algorithm where it has not found all occurrences. Model these iterations as coin flips, where the coin of an iteration turns up heads if this iteration finds no new occurrence. Heads comes up with probability \emph{at most} $1/2$ because each such iteration finds a new occurrence with probability at least $1/2$. By \Cref{obs:coins}, the probability that (for any $j$ in this sequence) after $j$ coin flips heads comes up $\log_2 j + \Theta(\log n)$ times in a row is at most $n^{-\Omega(1)}$. This situation is the only one in which the algorithm terminates before finding all occurrences.

\end{proof}

Hence, if we can find every occurrence that does not cross a cluster, we can find all occurrences with high probability. It remains to describe how to find these occurrences.

\subsubsection{Recovering All Occurrences for a Cluster} \label{sec:listing-recovering}

Every valid partial match of the root of the tree decomposition that does not map any vertex as \emph{unmatched} can be attributed to one or more subgraph isomorphisms. We construct these subgraph isomorphisms \emph{top down} while traversing the shortcut graph of valid partial matches in reverse order (only following edges that lead to a valid partial match). The algorithm keeps a set of current subgraph isomorphisms at every vertex in the graph and does a parallel BFS of limited depth. When visiting a new vertex of the shortcut graph (which contains a partial mapping $\phi$), every subgraph isomorphism in the list from the predecessor node is extended by $\phi$ and stored in the new vertex.

 As for the decision problem, we observe that only $k$ edges introduce a new vertex to the mapping. The other edges are shortcut so that overall at most $O(\log n)$ edges need to be traversed in between those $k$ edges. However, we now need to construct the possible subgraph isomorphism even through those shortcuts explicitly. Fortunately, as illustrated in \Cref{fig:trivial-match}, there is a unique way to extend a partial match through these shortcut edges, namely, do not change the current mapping at all. Hence, the overall depth of the reconstruction is $O(k \log ^2 n)$. 
 
By considering only occurrences that contain at least one vertex that is closest to the root of the BFS tree of the $k$-$d$ cover, every traversed path leads to at least one subgraph isomorphism, and the work is bounded by the size of all the subgraph isomorphisms. 

\subsection{Bounded Genus \& Apex-Minor-Free Graphs}\label{sec:apex-free}

Our results generalize to all (minor-closed) families of graphs where a bounded diameter graph has bounded treewidth. 
Observe that our treewidth $k$-cover algorithm from \Cref{sec:low-treewidth-cover} does not use anything specific to planar graphs. It outputs subgraphs of diameter $d$ that cover all occurrences of the pattern with constant probability. Moreover, our algorithm for bounded treewidth in \Cref{sec:bounded-treewidth-algo} only requires a treewidth decomposition of low width.
We start by giving the characterization of the graphs where our results hold and then discuss the few necessary changes.

\subsubsection{Locally Bounded Treewidth}

A family of graphs has \emph{locally bounded treewidth}~\cite{DBLP:journals/algorithmica/Eppstein00} if every graph of diameter $D$ has treewidth at most $f(D)$, for some function $f$. 
Surprisingly, all minor-closed families of graphs that have locally bounded treewidth have \emph{locally linear treewidth}~\cite{DBLP:conf/soda/DemaineH04}, meaning that a graph of diameter $D$ has treewidth $O(D)$.

The graphs of locally bounded treewidth have been characterized with respect to having certain \emph{excluded minors}. 
A graph $G$ that has a vertex $v$ that is connected to all other vertices in $G$ that becomes planar after removing $v$ is an \emph{apex-graph}. Such graphs do not have locally bounded treewidth. For example, consider the $n \times n$ grid with an additional vertex connected to all other vertices. This graph has diameter $2$, but because the grid has treewidth $n$~\cite{DBLP:journals/jct/RobertsonS86} this apex graph has treewidth at least $n$. Note that some apex graphs are planar (like the clique $K_4$) while others are not (like the clique $K_5$).

Interestingly, a minor-closed family of graphs of locally bounded treewidth must have an \emph{apex graph as an excluded minor}~\cite{DBLP:journals/algorithmica/Eppstein00} . For example, planar graphs exclude the apex graph $K_5$ as a minor (by Kuratowski's theorem~\cite{wilson1996introduction}). 
Examples of apex-minor-free graphs include bounded-genus-graphs.

\subsubsection{Parallel Tree Decomposition}

The missing piece to our parallel subgraph isomorphism algorithm on apex-minor-free graphs is a parallel tree decomposition algorithm. The algorithm from Lagergren~\cite{DBLP:journals/jal/Lagergren96} achieves poly-logarithmic depth for constant treewidth, but the depth of the algorithm is not polynomial in $\tau$. It becomes the bottleneck in our subgraph isomorphism algorithm.
\begin{theorem}[Lagergren~\cite{DBLP:journals/jal/Lagergren96}]
For a graph with treewidth $\tau$, computing a tree decomposition of width $8\tau+7$ takes $\tau^{O(\tau)}m$ work and $\tau^{O(\tau)}\log^3 n$ depth.
\end{theorem}
Together with the results from \Cref{sec:low-treewidth-cover} and \Cref{sec:bounded-treewidth-algo} this proves the generalized bounds. Similar results hold for disconnected patterns and listing all occurrences of the pattern.
\begin{theorem}
Deciding (with high probability) if a connected pattern graph $H$ occurs as a subgraph of an apex-minor-free graph $G$ takes $k^{O(k)} n \log^3 n$ work and $k^{O(k)} \log ^3 n$ depth.
\end{theorem}

\section{Planar Vertex Connectivity}

Vertex connectivity is a classic graph problem with applications in networking~\cite{DBLP:conf/soda/Censor-HillelGK14} and operations research~\cite{DBLP:journals/ipl/NagamochiII01}. Sequentially, $c$-vertex connectivity can be solved in linear time for planar graphs~\cite{DBLP:conf/soda/Eppstein95} and, more generally, in $O(c^2 n^2 \log n)$ time deterministically~\cite{DBLP:conf/focs/HenzingerRG96} and $O(m+c^{7/3}n^{4/3})$ time with high probability~\cite{DBLP:conf/stoc/NanongkaiSY19}. Recently, a sub-quadratic time deterministic algorithm~\cite{DBLP:journals/corr/abs-1910-07950} and a near-linear work~\cite{DBLP:journals/corr/abs-1905-05329} algorithm have been announced. 

 Two-connectivity and $3$-connectivity have long been solved (optimally) for general graphs with linear work and logarithmic depth~\cite{DBLP:journals/siamcomp/TarjanV85, DBLP:conf/stoc/MillerR87}. In contrast, no sub-quadratic work poly-logarithmic depth $4$-connectivity algorithm was available even for planar graphs prior to our work. 

We show that vertex connectivity can be solved with $O(n \log n)$ work and $O(\log^2 n)$ depth in planar graphs. This result is possible because the vertex connectivity is closely related to certain separating cycles in a target graph that is constructed based on a planar embedding of the original graph (details below). Moreover, we use that the work of our subgraph isomorphism algorithm is $O(n \log n)$ \emph{for any constant size pattern}. Eppstein~\cite{DBLP:conf/soda/Eppstein95} uses this idea (attributed to Nishizeki) for his \emph{sequential} linear work vertex connectivity algorithm. We describe the approach and the necessary changes to our parallel algorithm.

\subsection{From Connectivity to Separating Cycles}\label{sec:vertex-redux}

We show how to construct the target graph that we use to solve vertex connectivity, based on an idea attributed to Nishizeki~\cite{DBLP:conf/soda/Eppstein95}. See \Cref{fig:vertexcut} for an illustration.

Embed the graph $G$ in the plane. Use the embedding to construct a \emph{bipartite} target graph $G'$ from $G$ as follows. One side of the bipartite graph consists of the vertices from $G$. The vertices on this side are the \emph{original vertices}. The other side has a vertex $f$ for each face $f$ in the original graph $G$. The vertices on this side are the \emph{face vertices}. A face vertex $f$ of $G'$ and an original vertex $v$ of $G'$ are connected if and only if the face $f$ contains the vertex $v$ in the graph $G$. Observe that because the graph $G'$ is bipartite, all its cycles have even length. 

\paragraph{Separating Subgraphs}
A subgraph $H'$ of a graph $G$ \emph{separates the vertex set $S\subseteq V(G)$} if the graph $G[V(G) \backslash V(H')]$ we get from removing all vertices of $H'$ from $G$ contains at least two vertices from $S$ in two different connected components. 

\begin{figure}
	\includegraphics[width=1.0\linewidth]{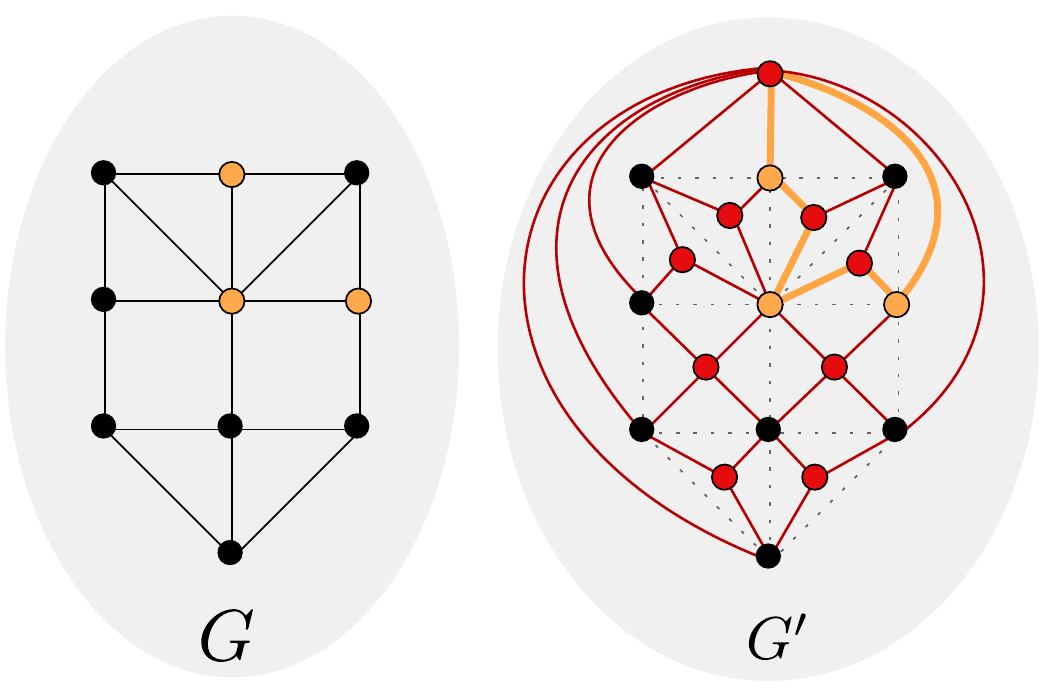}
\caption{
To construct the target graph $G'$ from the embedding of the graph $G$, place a vertex $v$ inside every face $f$ of $G$ and connect this vertex $v$ to all the vertices of the face $f$ (remove the original edges). Since there is a $6$-cycle (highlighted and bold) in $G'$ that separates the original vertices (black), but no smaller such cycle, the graph $G$ is $3$-connected. This cycle contains three original vertices (highlighted) whose removal disconnects the graph $G$.
}
\label{fig:vertexcut}
\end{figure}

\begin{lemma}[Nishizeki / Eppstein~\cite{DBLP:conf/soda/Eppstein95}] \label{lem:sepcycles}
If $G$ is $2$-connected and the shortest cycle in the bipartite graph $G'$ that separates the set of original vertices has length $2c$, then $G$ has vertex connectivity $c$.
\end{lemma}

This leads us to our algorithm to decide planar vertex connectivity in parallel. 
First, check if the graph is $2$-connected and if it is $3$-connected using existing algorithms~\cite{DBLP:journals/siamcomp/TarjanV85, DBLP:conf/stoc/MillerR87}. If the graph is $3$-connected, check if there is a cycle of length $8$ in $G'$ that separates the original vertices of $G'$. If so, the graph $G$ has vertex connectivity $4$. Otherwise, the graph $G$ has vertex connectivity $5$.

\begin{lemma}
Deciding Planar Vertex Connectivity (w.h.p) takes $O(n \log n)$ work and $O(\log ^ 2 n)$ depth.
\end{lemma}
\begin{proof}
	The algorithm is correct by Lemma \ref{lem:sepcycles} and the fact that the vertex connectivity of a planar graph is at most $5$. This follows from Euler's formula, which implies that every planar graph has a vertex of degree at most $5$~\cite{wilson1996introduction}. Removing the neighbors of this vertex disconnects the graph, hence the graph is not $6$-connected.

Constructing a planar embedding takes $O(n)$ work and $O(\log^2 n)$ depth~\cite{DBLP:conf/focs/KleinR86}. 
Together with the modifications described in \Cref{sec:separating-si} (Lemma \ref{lem:separating-si}) this implies the work and depth bounds.
\end{proof}

Hence, we need to augment our subgraph isomorphism algorithm so that it can find a \emph{subgraph that separates} a set of vertices (the original vertices in the case of the graph $G'$).

A simple approach to find all \emph{separating} cycles of a given length would be to enumerate \emph{all} cycles of a given length using the algorithm from \Cref{sec:listing-recovering} and check which are separating. However, there can be $\Theta(n^{4})$ many length $8$ cycles in a planar graph~\cite{DBLP:journals/jgt/HakimiS79}, so this would be too much work.

\subsection{Separating Subgraph Isomorphism}\label{sec:separating-si}

We generalize our parallel subgraph isomorphism algorithm so that it can find \emph{subgraphs that separate} a given set of vertices. Two modifications are necessary. These are similar to what was necessary for the sequential algorithm~\cite{DBLP:conf/soda/Eppstein95} for cycles. The first modification is to the parallel treewidth cover algorithm from \Cref{sec:low-treewidth-cover}. This modification ensures that a subgraph that is separating in the original graph is also separating in each of the graphs in the cover. The second modification concerns the algorithm for bounded treewidth subgraph isomorphism from \Cref{sec:bounded-treewidth-algo}. It extends the state space of the recursion to keep track of which vertices are separated by the subgraph and which can be in the same component after removing the subgraph.

\emph{$S$-Separating Subgraph Isomorphism} asks if there exists an occurrence $H'$ of the pattern graph $H$ in the target graph $G$ that separates the vertex set $S\subseteq V(G)$. If the pattern graph is a cycle, the problem is called \emph{$S$-Separating Cycle}.

\subsubsection{How to Modify the $k$-$d$-Cover}

Start by clustering the graph $G$ as usual. Then, for each cluster, merge all neighboring clusters into a single vertex each (do not choose these as the source for the BFS).
Then, in each cluster, instead of returning the graph $G_i$ (which is an induced subgraph of the cluster), merge all connected components of the cluster that result after removing $V(G_i)$ into a single vertex each. This produces a set of minors of the graph (instead of a set of induced subgraphs), as shown in \Cref{fig:cover-sep}.
 
 When proceeding to find an $S$-separating subgraph in these minors, consider each merged vertex that contains at least one vertex of the set $S$ to be in the set $S$. Moreover, do not allow the occurrence of the pattern to contain any of the merged vertices (the other vertices are in a set of allowed vertices $A$).

\begin{figure}\
	\includegraphics[width=1.0\linewidth]{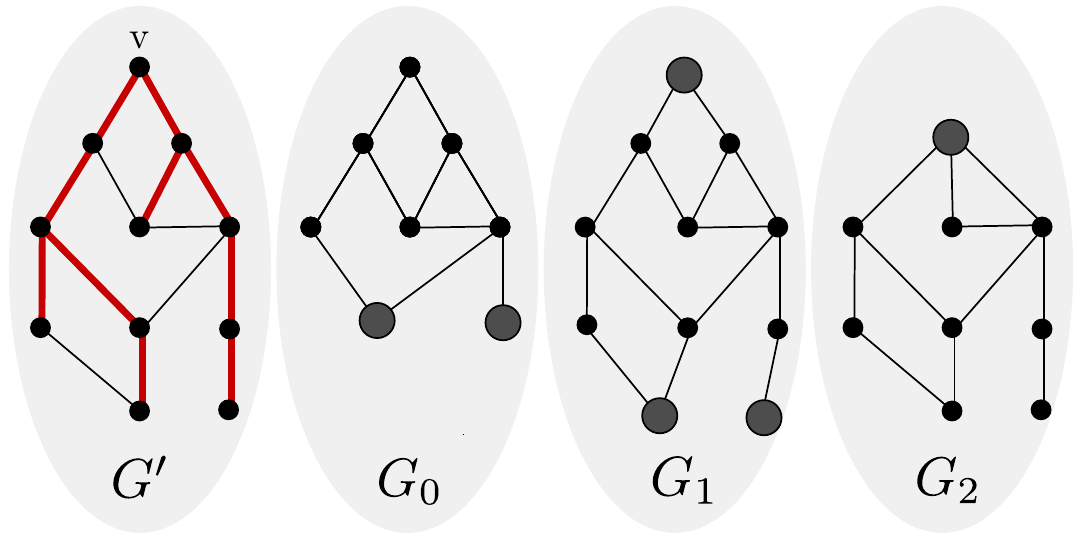}
	\vspace{-2em}
\caption{
In addition to the vertices from the $k$-$d$ cover, some vertices correspond to merged subgraphs (these are drawn larger in the picture).
A subgraph of diameter $d$ (here $d=2$) that is separating in $G'$ is separating in at least one of the minors $G_0$, $G_1$, $G_2$ using only the original (small) vertices.
}
	\vspace{-2em}
\label{fig:cover-sep}
\end{figure}

\subsubsection{How to Modify the Bounded Treewidth Algorithm}
The generalized algorithm must separate $S$ and only contain vertices from the set of allowed vertices $A$. 
To restrict the found occurrences to only contain vertices from the set of allowed vertices, it suffices to restrict the mapping at each tree decomposition node to $A$.

The idea to find an occurrence that separates $S$ is that we record which vertices are separated by the occurrence. Removing such an occurrence creates at least two connected components. We call one of these components the \emph{inside} vertices and the rest of the vertices the \emph{outside vertices}. Observe that after removing a separating occurrence from the graph, every resulting connected component must either consist of only inside or consist of only outside vertices.

We extend the construction of partial matches. A partial match for node $X$ has an additional set $I_X\subseteq X$ of \emph{on the inside} vertices and a set $O_X\subseteq X$ of \emph{on the outside} vertices.
Moreover, it has a boolean $i_x$ to keep track if any of the vertices in $S$ that occur in the subgraph induced by the current tree decomposition node are on the inside (and a boolean $o_x$ to store if any of those vertices are on the outside). This bookkeeping ensures that at least one vertex is on both sides -- otherwise, the subgraph would not be separating. 

We adapt the semantics of the combination rules accordingly to reflect the intuition that partial matches keep track of which vertices are on the inside or outside. Consider a node $X$ of the decomposition tree, one of its children $Y$, and the (extended) partial matches $(\phi_X, C_X, U_X, I_X, O_X, i_x, o_x)$ of $X$ and $(\phi_Y, C_Y, U_Y, I_Y, O_Y, i_y, o_y)$ of $Y$. Then, for the partial matches to be valid, ensure the following:
\begin{itemize}
	\item Every connected component of the subgraph of $G$ induced by the vertices in $X$ that are not mapped onto by the function $\phi_X$ is either fully in $O_X$ or fully in $I_X$. Similarly for $Y$. 
	\item The inside and outside of $X$ and $Y$ have to be consistent: For any vertex $u$, if $u \in X\cap Y$ then $u\in I_X$ if and only if $u\in I_Y$ and $u\in O_Y$ if and only if $u\in O_X$. 
	\item The parent match has to `remember' if any vertex is in $S$ and on the inside or outside. Specifically, for a vertex $u\in S$, $u \in I_X$ implies $i_x$ and $u\in O_X$ implies $o_x$. Moreover, $i_y$ implies $i_x$ and $o_y$ implies $o_x$.
\end{itemize}
Finally, a valid partial match at the root must separate $S$ (which means $i_x$ and $o_x$ are both true at the root).

\begin{lemma}\label{lem:separating-si}
	Deciding \emph{Planar $S$-Separating Subgraph Isomorphism} (w.h.p.) for a connected pattern graph with $k$ vertices takes $O(k\log ^2 n)$ depth and $O(2^{9k} (3k+1)^{3k+1} n \log n)$ work.\end{lemma}
\begin{proof}
Computing connected components and contracting the edges takes $O(n)$ work and $O(\log n)$ depth~\cite{DBLP:conf/ipps/Gazit91}. The number of states for the recursion increases by at most $2^{3k+3}$. Hence, the number of considered combinations with the children increases by at most $O(2^{9k})$ at every node.
\end{proof}

When $k$ is a constant, the algorithm takes $O(n \log n)$ work and $O(\log ^2 n)$ depth. In \Cref{sec:vertex-redux}, the only missing piece to solve planar vertex connectivity in $O(n \log n)$ work and $O(\log ^2 n)$ depth is to find $S$-Separating $8$-cycles, which we have just described how to solve in the stated bounds. 

\section{Conclusion and Future Work}

We presented a randomized algorithm to decide planar subgraph isomorphism in $O(n \log n)$ work and $O(\log ^2 n)$ depth for constant size patterns. We used this result for deciding planar vertex connectivity in the same parallel bounds.

There are many interesting avenues for future work. 
Although we could use our subgraph listing algorithm to \emph{count} the number of occurrences, this is not work-efficient as the runtime grows with the number of occurrences. The difficulty comes from the randomized way in which we cluster the graph to construct a $k$-$d$ cover. A deterministic parallel $k$-$d$ cover would solve this issue and yield a deterministic algorithm overall.

Reducing the work dependency on the size of the pattern $k$ could be an essential step in improving the practicality of the approach. There are indications that $2^{\Omega(k/\log k)}$ is a lower bound for the dependency on $k$ for any planar subgraph isomorphism algorithm with polynomial dependency in $n$~\cite{DBLP:conf/focs/FominLMPPS16}, but there remains room for improvement regarding the exponential dependency on $k$. Moreover, faster parallel algorithms for tree decomposition would directly improve our bounds for apex-minor-free graphs.

For planar vertex connectivity, we reduced the gap between the work of our algorithm and the best sequential algorithm to $O(\log n)$. It is natural to ask if it is possible to solve planar vertex connectivity in $O(n)$ work and poly-logarithmic depth. More generally, in light of the recently announced sequential near-linear time vertex connectivity algorithm for sparse graphs~\cite{DBLP:journals/corr/abs-1905-05329}, it might be interesting to see if we can solve vertex connectivity in sparse graphs in near-linear work and low depth.

\clearpage

\appendix

\section{Decomposing a Tree into Paths}\label{sec:bough-decomp}

We prove \Cref{lem:bough-decomposition} using expression tree evaluation techniques. This means that we transform the problem into a problem of evaluating an expression tree of suitable operations. To evaluate this expression tree efficiently, we need to decompose the operations into unary functions satisfying certain properties, as described below.

Recall that the Lemma requires the tree to be split into $O(\log n)$ layers each consisting of disjoint paths. The idea is to compute for each vertex in the tree the layer in which the vertex occurs. This computes for each node a \emph{layer number}, where the layer number of the leafs is zero and the layer number of nodes closer to the root is monotonically increasing (as detailed below).

Each layer (i.e. subgraph induced by vertices with the same layer number) consists of a forest where each connected component is a path. Hence, it is easy to find and order these paths (using list ranking) once we have the layer numbers.

Next, we describe the recursive function $L$ that computes the layer numbers.
In a general rooted tree, the parent $b$ has the same layer number $l(b)$ as the maximum layer number of any of its children $a_1, \dotsc a_k$ if this maximum is unique (i.e., only one child has this layer number). Otherwise, the layer number of the parent is one larger than that maximum. In summary, the layer number $l(b)$ of node $b$ with children $a_1, \dotsc a_k$ with layer numbers  $l_1, \dotsc l_k$ is given recursively:
\begin{align*}
L (l_1, \dotsc l_k) =
\begin{cases}
	{\max (l_1, \dotsc l_k) } & \hfill \text{if the maximum is unique ;}\\
	\max (l_1, \dotsc l_k)+1 & \hfill \text{otherwise .}
\end{cases}
\end{align*}
The layer number of a leaf is $0$. This recursive description works because the case where the maximum is unique corresponds to when the parent is part of the same path as the child that obtains this maximum. If two children have the same layer number, the parent must start its own path and a new layer.

 Moreover, observe that it becomes clear why there are $O(\log n)$ layers: For a parent to have a larger layer number than one of its children, there need to be at least two children of the same maximal layer number. This means that the number of nodes in a layer decreases by at least a factor $2$ when going to a higher layer.
 
 We proceed to describe the conditions for applying the efficient tree contraction based expression tree evaluation techniques, as summarized in \Cref{lem:expression}.
A family of unary functions is \emph{closed under composition} if the composition of any two functions in the family is also in the family.
A family of unary functions $\mathcal{F}$ over the domain $\mathcal{D}$ is \emph{closed under projection} with respect to a $k$-ary function $h:\mathcal{D}^k \rightarrow \mathcal{D}$ if for all tuples $a_1, \dotsc, a_{k-1} \in \mathcal{D}^{k-1}$ and all indexes $i$ (between $1$ and $k$) the function $h(a_1, \dotsc, a_{i-1}, x, a_{i+i} \dotsc, a_{k-1}): \mathcal{D} \rightarrow \mathcal{D}$ (a unary function of $x$) is in the family $\mathcal{F}$.

\begin{lemma}\label{lem:expression}
	If there is a family of $O(1)$-computable functions that is closed under composition and closed under projection with respect to all the operations in an expression tree of $n$ nodes, then evaluating the expression tree takes $O(n)$ work and $O(\log n)$ depth~\cite{Reif:1993:SPA:562546}.
\end{lemma}
The intuition is that the expression tree evaluation repeatedly contracts the expression tree. For this procedure to be well-defined, the algorithm needs to express partially evaluated subtrees using these unary functions.
Next, we exhibit such a suitable family of unary functions for the function $L$ that maps the layer number of the children to the layer number of the parent. 
%

We define a set of unary functions over the domain of natural numbers, where for each natural number $i$, there are two functions: a function $f_i^{\neq}(x)$ and a function $g_i^{=}(x)$. Intuitively, the functions $f_i^{\neq}(x)$ record a state where the maximum (so far) is unique and equal to $i$. The functions $g_i^{=}(x)$ record the state where the maximum is not unique and equal to $i$. Formally, we set:
\begin{align*}
f_i^{\neq}(x) &=
\begin{cases}
	i+1 & \hfill \text{if $i=x$ ,}\\
	\max(i, x) & \hfill \text{otherwise .}
\end{cases}\\
g_i^{=}(x) &=
\begin{cases}
i+1 & \hfill \text{if $i \geq x$ ,}\\
x & \hfill \text{if $i<x$ .}
\end{cases}
\end{align*}

\noindent
We check that the function class is closed under composition. For any natural numbers $i$ and $j$, the following holds:
\begin{align*}
g_j^{=}( \ f_i^{\neq}(x) \ ) = f_i^{\neq}( \ g_j^{=}(x) \ ) &=
\begin{cases}
g_i^{=}(x) & \hfill \text{if $i=j$ ,}\\
f_i^{\neq}(x) & \hfill \text{if $i>j$ ,}\\
g_j^{=}(x) & \hfill \text{if $j>i$ .}
\end{cases}\\
f_i^{\neq}( \ f_j^{\neq}(x) \ ) &=
\begin{cases}
g_i^{=}(x) & \hfill \text{if $i=j$ ,}\\
f_{\max(i, j)}^{\neq}(x) & \hfill \text{otherwise .}\\
\end{cases}\\
g_i^{=}( \ g_j^{=}(x) \ ) &= \ \
g_{\max(i, j)}^{=}(x) & \hfill \text{}\\
\end{align*}

To check that the function class is closed under projection with respect to $L$, consider a sequence of layer values $\mathcal{L}=l_1, \dotsc, l_{k-1}$. Let $l_{\max}$ be the maximum of $\mathcal{L}$. For any valid index $i$ we have that:
\begin{align*}
L (l_1, \dotsc, l_{i-1}, x, l_{i+1}, \dotsc, l_{k-1}) = 
\begin{cases}
f_{l_{\max}}^{\neq}(x) & \hfill \text{if $l_{\max}$ is unique in $\mathcal{L}$,}\\
g_{l_{\max}}^{=}(x) & \hfill \text{otherwise.}
\end{cases}
\end{align*}

\balance

\bibliographystyle{plain}
\bibliography{isomorphs,isomorph-apps}

\end{document}